\documentclass{ifacconf}

\usepackage{graphicx}      
\usepackage{natbib}        
\usepackage{tikz}
\usetikzlibrary{positioning}
\usepackage{circuitikz}
\usepackage{amsmath}
\usepackage{amsmath,amssymb,amsfonts,theorem}

\newtheorem{assumption}{Assumption}
\newtheorem{definition}{Definition}

\newtheorem{lemma}{{\bf Lemma}}

\newtheorem{corollary}{{\bf Corollary}}
\newtheorem{proposition}{{\bf Proposition}}

\newcommand{\R}{\mathbb{R}}
\begin{document}
\begin{frontmatter}

\title{Astrocytic gliotransmission as a pathway for stable stimulation of post-synaptic spiking: Implications for working memory \thanksref{footnoteinfo}} 

\thanks[footnoteinfo]{The work of M. Jafarian is supported by the Marie Skłodowska-Curie Fellowship, project ReWoMeN.}

\author[First]{Valentin Würzbauer}
\author[Second]{Kerstin Lenk}
\author[First]{Matin Jafarian}

\address[First]{Delft Center for Systems and Control, Delft University of Technology, The Netherlands. E-mail: v.wurzbauer@student.tudelft.nl; m.jafarian@tudelft.nl}
\address[Second]{Institute of Neural Engineering, Graz University of Technology, Austria. E-mail: kerstin.lenk@tugraz.at}

\begin{abstract}                
The brain consists not only of neurons but also of non-neuronal cells, including astrocytes. Recent discoveries in neuroscience suggest that astrocytes directly regulate neuronal activity by releasing gliotransmitters such as glutamate. In this paper, we consider a biologically plausible mathematical model of a tripartite neuron-astrocyte network. We study the stability of the nonlinear astrocyte dynamics, as well as its role in regulating the firing rate of the post-synaptic neuron. We show that astrocytes enable storing neuronal information temporarily. Motivated by recent findings on the role of astrocytes in explaining mechanisms of working memory, we numerically verify the utility of our analysis in showing the possibility of two competing theories of persistent and sparse neuronal activity of working memory.
\end{abstract}

\begin{keyword}
Control in neuroscience; Stability of nonlinear systems; Networked systems 
\end{keyword}

\end{frontmatter}
\section{Introduction}
While the role of neurons as the key component in the nervous system has been the subject of tremendous research since Cajal has drawn the first brain cells in the late 19th century, the function of astrocytes as an active partner in neural signaling pathways has only become evident in the last two decades \citep{llinas2003contribution}.  The concept of the tripartite synapse including an astrocyte as well as a pre- and a postsynaptic neuron as an extension to the classical neuron-to-neuron communication has been introduced in \citet{araque1999tripartite}. Until this point, glial cells, the class of cells to which astrocytes are assigned, have only been considered supportive cells that ensure the nutrition and structural support of neurons. 

Communication between neurons occurs in the form of electrical/chemical signals (spiking) via the so-called synapses. Chemical synapses are often enwrapped or closely contacted by astrocytes. Astrocytes’ primary signal for information transfer is calcium. When the astrocyte’s calcium level elevates, it can release transmitter molecules that directly act on neurons. In fact, a sequence of  action potentials triggers the release of messenger substrates of the presynaptic neuron. The neurotransmitters such as glutamate in the synaptic cleft activate membrane channels of both the postsynaptic neurons triggering a postsynaptic potential and the adjacent astrocyte. Then, the astrocytes start a cascade reaction that leads to calcium ($Ca^{2+}$) elevations. The astrocyte dynamics seemingly differ significantly from the neuronal dynamics with respect to space and time. While an action potential and the corresponding synaptic neurotransmission occur in the range of a few milliseconds, the elevation of intracellular $Ca^{2+}$ concentration happens within several seconds. Experiments have shown the release of gliotransmitters by astrocytes in response to high $Ca^{2+}$ levels affecting both the presynaptic and the postsynaptic neuron of the synapse \citep{savtchouk_gliotransmission_2018}. It has to be mentioned that gliotransmission in physiological astrocytes is still a matter of debate today \citep{fiacco2018multiple, savtchouk_gliotransmission_2018}. Recently, \citet{gordleeva2021modeling} and \citet{de2022multiple} have shown that neuron-astrocyte network models with varying biological accuracy are able to store neuronal stimulation via an astrocyte mechanism that enhances synaptic efficacy.

The above findings have motivated the investigation of the role of astrocytes in regulating the neural mechanisms underlying brain functioning e.g. in Parkinson's and epilepsy, as well as cognitive functions such as working memory \citep{gv20,de2022multiple}. Working Memory (WM) is a general-purpose cognitive system responsible for temporarily processing information in service of higher-order cognition such as reasoning and decision-making. 
The neuronal activity within a limited temporal interval is assumed to form the mechanism of storing in formation in the WM \citep{adamsky2018astrocytes}. In fact, several theories exist about the underlying mechanism that stores information. These theories can be divided into two main classes of persistent and sparse neuronal activities \cite{barak2014working}. Both, persistent activity as well as sparse activity, have been observed in WM experiments with primates \citep{funahashi1989mnemonic,lundqvist2018working}. The aforementioned slow reaction of astrocytic gliotransmission is an interesting theory of WM. The recent findings of the regulatory role of astrocytes in WM networks have shed new light on the debates between persistent and sparse activities \citep{gordleeva2021modeling, de2022multiple}. 

In this paper, we study the stability of the nonlinear astrocyte dynamics in a tripartite model and show its effects on modulating the spiking rate of the postsynaptic neuron. We will also numerically verify the obtained insights in a large-scale WM network. Our results show that astrocytes can enable both sparse and persistent neural activities for the WM network.

Compared with the literature, our contribution is twofold. First, different from \citet{gordleeva2021modeling} and \citet{de2022multiple} where the effect of gliotransmission towards the presynaptic neuron has been studied, we consider the integration of slow inward current to the postsynaptic neuron and the examination of persistent and sparse firing depending on the strength of gliotransmission. Second, we perform a stability analysis for a detailed biologically plausible nonlinear model of an astrocyte considered in \citet{gordleeva2021modeling}. We show that the nonlinear model is a positive system and it is ultimately bounded. We charatrize the ultimate bound and show the existence of a locally asymptotic stable equilibrium in the positive orthant. To the best of our knowledge, stability analysis of astrocyte dynamics has only been considered in \citet{de2022multiple} which studied the presynaptic stimulation effects using linearization of a model that is less biologically detailed than the model we consider. Furthermore, we provide numerical results to indicate the implications of our results in supporting the possibility of the co-existence of both sparse and persistent neural activities in WM.

The paper is organized as follows. Section~\ref{sec:2} reviews the tripartite model \citep{gordleeva2021modeling}, and presents our spiking rate model and assumptions. In Section~\ref{sec:3}, both stability analysis and its corresponding numerical validation are presented. A large-scale simulation shows the performance of a neuron-astrocyte network performing WM tasks in Section~\ref{sec:4}. Finally, the paper is concluded in Section~\ref{sec:5}.\\[1mm]
\noindent{\bf{Notations}:} Let $\R_+ = [0, \infty)$, and $\R^{n}_{+}$ is the set of $n$-tuples for which all components belong to $\R_+$. Denote the boundary of $\R^{n}_{+}$ by $bd(\R^n_+)$. 
\begin{definition}[Positive systems]
System $\dot x=f(x(t))$ is positive if and only if $\R^{n}_{+}$ is forward invariant \citep{de2001stability}.
\end{definition}
\section{Modeling} \label{sec:2}
We consider a tripartite synapse including a presynaptic and a postsynaptic neuron as well as an astrocyte as shown in Fig.~\ref{fig:blockalg}. We study the neuron-astrocyte network model by \citet{gordleeva2021modeling} which combines the neuron model by \citet{izhikevich2003simple} and the astrocyte models by \citet{nadkarni2003spontaneous} and \citet{ullah2006anti}.

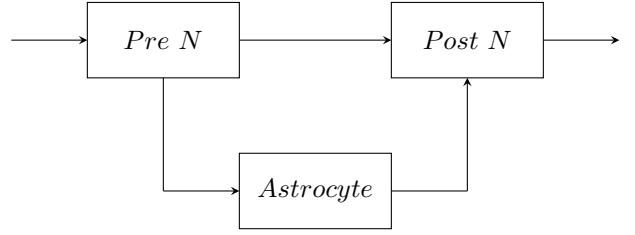
\begin{figure}[h]
    \centering
    \begin{tikzpicture}
        \draw[-stealth] (0,0) -> (1,0);
        \draw (1,-0.5) rectangle (3,0.5)
        node[midway]{$Pre\ N$};
        \draw[-stealth] (3,0) -- (5,0);
        \draw (5,-0.5) rectangle (7,0.5) node[midway]{$Post\ N$};
        \draw[-stealth] (7,0) -> (8,0);
        \draw (3,-2.5) rectangle (5,-1.5) node[midway]{$Astrocyte$};
        \draw[-stealth] (2,-2) -- (3,-2);
        \draw (2,-2) -- (2,-0.5);
        \draw (5,-2) -- (6,-2);
        \draw[-stealth] (6,-2) -- (6, -0.5);
    \end{tikzpicture}
    \caption[Block diagram of the tripartite synapse]{Block diagram representation of a tripartite synapse: The presynaptic neuron (Pre N) signals to the postsynaptic neuron (Post N) and the astrocyte. The astrocyte can release gliotransmitters towards the postsynaptic neuron.}
    \label{fig:blockalg}
\end{figure} 

\subsection{Neuronal Dynamics}
The Izhikevich neuron model \citep{izhikevich2003simple} in combination with the glutamate release dynamics proposed in \citet{gordleeva2012bi} is given by
\begin{equation}\label{eq:neuron}
    \begin{array}{l}
    {\dot V}=0.04 V^2+5 V-U+140+I_{\mathrm{app}}+I_{\mathrm{syn}}+I_{\mathrm{astro}},  \\\\
    {\dot U}=a\left(b V-U\right), \\\\
    {\dot G}=-\alpha_{\mathrm{glu}} G+k_{\mathrm{glu}} \Theta\left(V-30 m V\right),
    \end{array}
\end{equation}
with $V$, $U$, and $G$ denoting the membrane potential, membrane potential's  recovery variable, and the released glutamate in the synaptic cleft where $\Theta(x)$ is the Heaviside step function. Additionally, if $V  \geq 30 \mathrm{mV}$, then the neuron spikes, and $V$ and $U$ are updated to $c$, and $U+d$ respectively. The constant parameters $a$, $b$, $c$, and $d$ are chosen to resemble fast-spiking behavior as it is mostly found in the prefrontal cortex. The currents $I_{app}$, $I_{syn}$, and $I_{astro}$ denote an external current, the synaptic current via neurotransmission, and the astrocytic current via gliotransmission.

\subsection{Astrocytic Dynamics}
We consider a three-state single astrocyte model given by \citet{li1994equations} with extensions by \citet{nadkarni2003spontaneous} and \citet{ullah2006anti}. The model states are $[IP_3]$, $[Ca^{2+}]$, and $h$. The first two describe the intracellular concentration of $IP_3$, which is a second messenger substrate, and $Ca^{2+}$ ions. The parameter $h$ denotes the share of active $IP_3$ receptors connecting the endoplasmic reticulum with the intracellular space. Denoting $[IP_3]$, $[Ca^{2+}]$, and $h$ with $x_1$, $x_2$, and $x_3$ respectively, the astrocyte dynamics \citep{gordleeva2021modeling} obey
\begin{equation}\label{eq:ast1}
    \begin{array}{l}
    {\dot x}_1=\frac{x_1^*-x_1}{\tau_{I P 3}}+\frac{v_4(x_2+(1-\alpha)k_4)}{x_2+k_4}+u(t), \\\\
    {\dot x}_2=-k_1 x_2 + c_1v_1(x_1x_2x_3)^{3} \frac{\frac{c_0}{c_1}-(1+\frac{1}{c_1})x_2}{(x_1+d_1)^{3} (x_2+d_5)^{3}}\\
    \hspace{10mm}-\frac{v_3\ x_2^{2}}{k_3^2+x_2^2} + \frac{v_6\ x_1^{2}}{k_2^2+x_1^2} + c_1v_2(\frac{c_0}{c_1}-(1+\frac{1}{c_1})x_2),\\
    {\dot x}_3=a_2\left(d_2 \frac{x_1+d_1}{x_1+d_3}\left(1-x_3\right)- x_2 x_3\right),
    \end{array}
\end{equation}
where all the coefficients are positive numbers provided in Appendix~\ref{sec:A1}. The system's input, connecting the presynaptic neuron and the astrocyte, is the $IP_3$ production induced by presynaptic glutamate release,  i.e., $u(t)= J_{\mathrm{glu}}$(t), given by
\begin{equation} \label{eq:J_glu}
J_{\mathrm{glu}}= A_{\mathrm{glu}} \Theta (G-G_{thr}), 
\end{equation}
where exceeding a threshold value $G_{\mathrm{thr}}$ of the presynaptic neuron provokes $J_{glu}$\footnote{The condition for glutamate-induced $IP_3$ production $J_{glu}$ for the analysis of the tripartite synapse as described in Equation \eqref{eq:J_glu} is modeled slightly different from \citet{gordleeva2021modeling}}. Finally, the inward current of the postsynaptic neuron under the effect of gliotransmission, $I_{astro}$, is modeled as described by \citet{nadkarni2003spontaneous}. Define $y=x_2/\mathrm{nM}-196.69$, then
\begin{equation} \label{eq:I_astro}
    I_{\text {astro }}=2.11 \Theta(\ln y) \ln y.
\end{equation}
\subsection{Complexity Reduction and Assumptions}
We are interested in the stability analysis of the interconnected model as in Fig.~\ref{fig:blockalg}, where the pre-synaptic neuron provides input to the two other blocks. As shown, the astrocyte dynamics are nonlinear as well as the neuron dynamics. To reduce the complexity of neuronal dynamics, we consider the dynamics of the firing rate of the post-synaptic Izhikevich neuron denoted by $x_4$. We model the firing rate dynamics as 
\begin{equation} \label{eq:x_4}
    \dot{x}_4=-x_4+f(\eta {I}_{astro}-I_{thr}),
\end{equation}
where $f(.)$ is an odd function, $f(0)=0$, $I_{astro}$ represents the astrocytic current, and $I_{thr}$ is a threshold value. The parameter $\eta$ is introduced in order to differentiate between weak and strong gliotransmission \citep{de2011tale}. The derivation of the function $f(\cdot)$ based on the Izhikevich neuron is provided in more detail in Appendix~\ref{sec:A2}.

Moreover, we impose the following mild assumption on the input signal $G_{glu}$ of the astrocyte.
\begin{assumption}
The input signal $J_{glu}(G)$ is a smooth function, such that $0 \leq J_{glu} \leq A_{glu}, \forall G$.
\end{assumption}
\section{Stability Analysis} \label{sec:3}
In this section, the effect of gliotransmission on the postsynaptic neuronal firing rate is studied. First, we focus on the stability the astrocyte dynamics. Second, we investigate the effects of astrocyte output, its second state, on the firing rate of the post-synaptic neuron. We then verify the results via numerical simulation of the extended model, composed of the astrocyte and the post-synaptic neuron, as well as the tripartite model. 
\subsection{Stability Analysis of the Astrocyte Dynamics} \label{sec:3.1}
Define $x=[x_1\ x_2\ x_3]^\top$. The astrocyte dynamics given in \eqref{eq:ast1} admits the following general representation
\begin{equation}\label{eq:ast2}
\dot x= f(x)+ b + B u(t),    
\end{equation}
where $b=[\frac{x_1^*}{\tau_{I P 3}}\  v_2 c_0\  0]^T$, $B=[1\ 0\ 0]^T$, and $f(x)$ captures all nonlinear terms in the dynamics of $x_1,x_2,x_3$ as presented in \eqref{eq:ast1}. Since the input to the system \eqref{eq:ast1} obeys Assumption $1$, we consider it as a bounded positive additive term. We use the following Lemma derived from Property \eqref{p} in \citep{de2001stability} to show the positivity of the astrocyte dynamics.
\begin{lemma} 
System \eqref{eq:ast2} with the input based on Assumption $1$ is positive if and only if 
\begin{equation}\label{p} 
{\bf{P}}: \forall x \in bd(\R^n_+): x_i = 0 \Rightarrow f_i(x) \geq 0.
\end{equation}
\end{lemma}
\begin{proposition}\label{pr1}
The nonlinear system in \eqref{eq:ast1} is positive. That is, $\forall x(0) \in \R^n_+$ it holds that $x(t) \in \R^n_+$.
\end{proposition}
\begin{proof}
We verify that system \eqref{eq:ast1} satisfies the property in Lemma $1$. Since all coefficients in \eqref{eq:ast1} are positive, substituting $x_2=0$ in \eqref{eq:ast1} gives ${\dot x}_2 >0$. Thus, $x_2(t) \geq 0$ holds. Substituting $x_1=0$ and $x_2 \geq 0$ in $x_1$ dynamics, gives ${\dot x}_1 >0$, hence $x_1(t) \geq 0$. Similarly, we conclude that $x_3(t) \geq 0$, which ends the proof.  
\end{proof}
We now continue by proving uniform ultimate boundedness (Definition 4.6, and Theorem 4.18 \citep{khalil2015nonlinear}) for system \eqref{eq:ast1}. We benefit from the positivity of the system in conducting the proof.
\begin{proposition}\label{pr2}
Consider astrocyte dynamics in \eqref{eq:ast1} with $x(0) >0$, under Assumption~$1$. The system is uniformly ultimately bounded, and its solution converges to the set $$\Omega=\{x \in {\mathcal R}^{+}_{3}: 0 \leq x_1 \leq \mu_1; 0 \leq x_2 \leq \mu_2; 0 \leq x_3 \leq 1\},$$ with $\mu_1=x_1^*+\tau_{I P 3}(v_4+A_{\mathrm{glu}})$ and $\mu_2=\frac{v_6+ c_0 (v_1-v_2)}{k_1+v_2(1+c_1)}$.
\end{proposition}
\begin{proof}
Let us first rewrite the dynamics in \eqref{eq:ast1} as follows
\begin{equation}\label{eq:system}
    \dot{x}=   
    \left[\begin{array}{c}
    \frac{x_1^*- x_1}{\tau_{I P 3}}+ N_1 +J_{\mathrm{glu}} \\
    - \beta_1 x_2 -\beta_2 + N_2 (\frac{c_0}{c_1}-(1+\frac{1}{c_1})x_2)- N_3 + N_4\\
    a_2 (N_5(1-x_3)-x_2 x_3)
    \end{array}\right]
\end{equation}
where $\beta_1=k_1+v_2(1+c_1)$, $\beta_2=c_0 v_2$, and each $N_i$ denotes a bounded nonlinear term in \eqref{eq:ast1}. Now, consider the Lyapunov function $V(x)= \frac{1}{2} (x_1^2+x_2^2+x_3^2)$. First, define $V_3= \frac{1}{2} x_3^2$, where ${\dot V}_3= x_3 {\dot x}_3= a_2 N_5 x_3 (1-x_3)-a_2x_2x_3^2$. Since $0< \frac{d_1 d_2}{d_3} \leq N_5 \leq d_2$, and based on Proposition \ref{pr1},  $x_2 \geq 0, x_3 \geq 0$ hold, we conclude that ${\dot V}_3 <0$ if $x_3 >1$. Thus, $x_3$ converges to the interval $[0,1]$. We continue by computing the bounds of each $N_i$, considering $x_1, x_2 \geq 0, 0\leq x_3 \leq 1$. We obtain $(1-\alpha) v_4 \leq N_1 \leq v_4; 0 \leq N_2 \leq c_1 v_1; 0 \leq N_3 \leq v_3; 0 \leq N_4 \leq v_6$. Computing the derivative of $V_{1,2}(x)= \frac{1}{2} (x_1^2+x_2^2)$, gives 
\begin{equation}
    \begin{array}{l}
    {\dot V}_{1,2} \leq -\frac{x_1^2}{\tau_{I P 3}} + x_1(\frac{x_1^*}{\tau_{I P 3}}+v_4+J_{\mathrm{glu}})\\\\
    \hspace{8mm} - \beta_1 x_2^2 +(v_6+c_0(v_1-v_2)) x_2.
    \end{array}
\end{equation}
So, if $x_1 > x_1^*+\tau_{I P 3}(v_4+J_{\mathrm{glu}})$ and $x_2 > \frac{v_6+ c_0 (v_1-v_2)}{k_1+v_2(1+c_1)}$, then ${\dot V}_{1,2} <0$ holds. The latter two conditions together with $x_3 >1$ guarantee that $\dot V\leq 0$ holds, since $V=V_{1,2}+V_3$. Thus, the solution of the system converges to the set $\Omega$ which ends the proof.
\end{proof}
As we discussed in the proof of Proposition \ref{pr2}, the nonlinear terms in \eqref{eq:ast1} are bounded. Hence, system \eqref{eq:ast2} can be approximated by the following linear system
\begin{equation}\label{lin} 
\dot x \approx  C x+ d + b+ B u(t),
\end{equation} 
where $B$ and $b$ are as defined for \eqref{eq:ast2}, and $d$ represents a constant vector approximating the nonlinear terms. Notice that matrix $C$ is a negative definite diagonal matrix, which guarantees local asymptotic stability of the equilibrium of the system \eqref{lin}. 
\begin{corollary} 
There exists a locally asymptotically stable equilibrium belonging to $\R^n_+$ for the system \eqref{eq:ast2}.
\end{corollary}
Numerical examples corresponding to the above result are given in the Appendix~\ref{sec:A3}.
\subsection{Stability of firing rate of the post-synaptic neuron}
We now continue by connecting the astrocyte output to the input of the postsynaptic neuron (as shown in Fig.~\ref{fig:blockalg2}). 
\begin{figure}[h]
    \centering
    \begin{tikzpicture}
        \draw[-stealth] (0,0) -> (1,0) node[above,xshift=-1cm]{$u = J_{glu}$};
        \draw (1,-0.5) rectangle (3,0.5) node[midway]{$Astrocyte$};
        \draw[-stealth] (3,0) -> (4,0) node[above,xshift=-0.5cm]{$x_2$};
        \draw (4,-0.5) rectangle (6,0.5) node[midway]{$Neuron$};
        \draw[-stealth] (6,0) -> (7,0) node[above,xshift=0cm]{$y = x_4$};
    \end{tikzpicture}
    \caption[Block diagram of the extended astrocyte model]{Block diagram of the extended astrocyte model}
    \label{fig:blockalg2}
\end{figure}
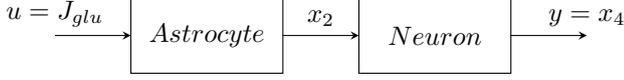
Based on the dynamics of $x_4$ in \eqref{eq:x_4}, since for $\eta I_{astro} < I_{thr}$, the map $f(.)$ gives zero, then $x_4=0$ is asymptotically stable. Also, $x_4$ is bounded for bounded $I_{astro}$, which is a function of $x_2$,  hence $x_4$ is input-to-state stable \citep{khalil2015nonlinear}. In particular for constant $\eta I_{astro} > I_{thr}$, $x_4$ converges to an asymptotically stable equilibrium in ${\R}_{+}$. The following Corollary summarizes these statements.
\begin{corollary}
 
The firing-rate of the post-synaptic neuron in \eqref{eq:x_4} is input-to-state stable. Moreover, there exists an asymptotically stable equilibrium in ${\R}_{+}^{n}$ for the extended system composed of \eqref{eq:x_4} and \eqref{eq:ast2}.
\end{corollary}

\subsection{Numerical Analysis} \label{sec:3.2}
In this section, we consider three cases to account for various possible scenarios. 
\begin{itemize}
    \item Case 1: No stimulation of astrocyte, i.e., $J_{glu} = 0$
    \item Case 2: Stimulation of astrocyte, i.e., $J_{glu} = 5 \frac{\mu M}{s}$ and strong gliotransmission, i.e., $\eta = 100 \%$ 
    \item Case 3: Stimulation of astrocyte, i.e., $J_{glu} = 5 \frac{\mu M}{s}$ and weak gliotransmission, i.e., $\eta = 25 \%$
\end{itemize}
The two scenarios of a persistently stimulated (Case 2,3) and an unstimulated astrocyte (Case 1) can be achieved by simulating the glutamate dynamics for the different spiking activity of the pre-synaptic neuron (see dynamics in Equation \eqref{eq:neuron}). The results of Proposition 2 and Corollary 1 are verified below. The simulation results show the ultimate boundedness property for all bounded stimulation cases. For the persistent and constant stimulation (Case 2), the states converge to an equilibrium. The short-period pulse of 0.2 seconds ($J_{glu}$) - which is selected in accordance to real experiments with primates \citep{funahashi1989mnemonic} - leads to bounded astrocytic response ($Ca^{2+}$ and $I_{astro}$) within the first 4 seconds. Additionally, a simulation without any input is performed to confirm a positive equilibrium value of the unforced system.

\begin{figure}[h]
    \centering
    \includegraphics[width = 0.99\linewidth]{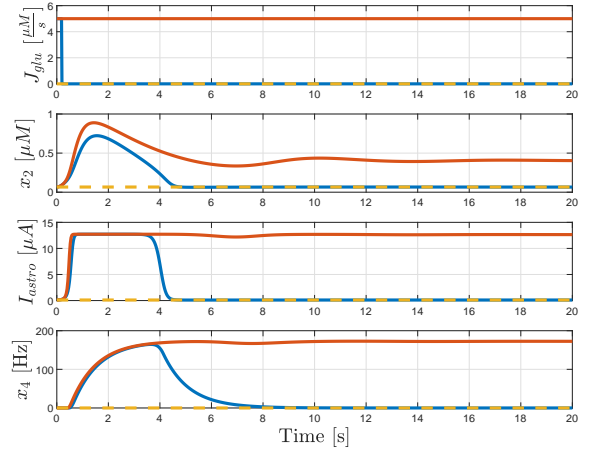}
    \caption[Numerical results for extended astrocyte dynamics]{The extended astrocyte model (Fig. 2) with different glutamate inputs: no input signal (yellow), a persistent input $J_{glu} = 5 \frac{\mu M}{s}$ (red) and a short input (blue).}
    \label{fig:extended_astrocyte1}
\end{figure} 

A second simulation is realized in order to examine and confirm the effect of weak and strong gliotransmission. Fig.~\ref{fig:extended_astrocyte2} shows that both Cases 2 and 3, have identical astrocytic behavior, as expected for the same input, and an entirely different postsynaptic firing frequency. The weak gliotransmission results in no visible neuronal postsynaptic activity although it still causes $I_{astro}>0$ toward the postsynaptic neuron. The reason for this observation is the type 2 neuron dynamics of the Izhikevich model which denotes a class of neurons that show a sudden high firing rate after exceeding an input threshold and no firing rate below the threshold. The weak astrocytic current does still lead to the easier onset of firing in the presence of additional inputs in a more realistic setting. 

\begin{figure}[h]
    \centering
    \includegraphics[width = 0.99\linewidth]{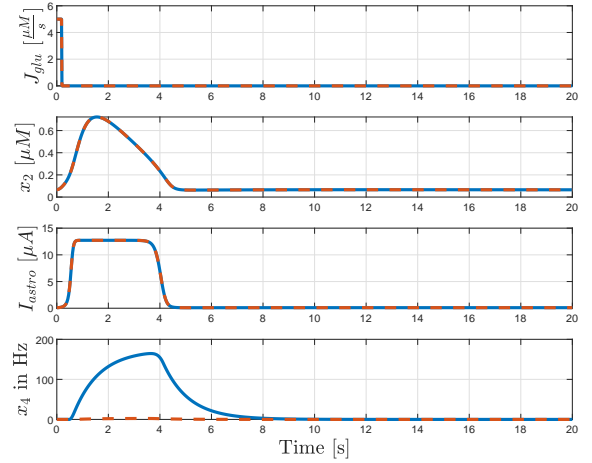}
    \caption[Numerical results for extended astrocyte dynamics]{The input and response of the model in Fig. 2 with: a short impulse signal with weak (red) and strong (blue) gliotransmission.
    }
    \label{fig:extended_astrocyte2}
\end{figure} 

Finally, the actual tripartite model as described in section~\ref{sec:2} is simulated. In this simulation, the whole nonlinear dynamics of neurons and the astrocyte are used. Fig.~\ref{fig:sim_tripartite} shows the firing frequency of the presynaptic and the postsynaptic neuron as well as the astrocytic $Ca^{2+}$ dynamics for the scenario with a short and a persistent input signal of $I_{app} = 100 \mu A$ (see~\eqref{eq:neuron}) for strong gliotransmission. Although the presynaptic neuron is spiking for the initial stimulation phase, the synaptic connection is not strong enough to initiate the firing of the postsynaptic neuron. Strong astrocytic current $I_{astro}$ induced by the enhanced $Ca^{2+}$ concentration within the astrocyte leads to temporal postsynaptic neuronal activity.
 All simulations are conducted using a Runge-Kutta-4 algorithm with a fixed time step $\Delta t = 0.1 ms$. 

\begin{figure}[h]
    \centering
    \includegraphics[width = 0.99\linewidth]{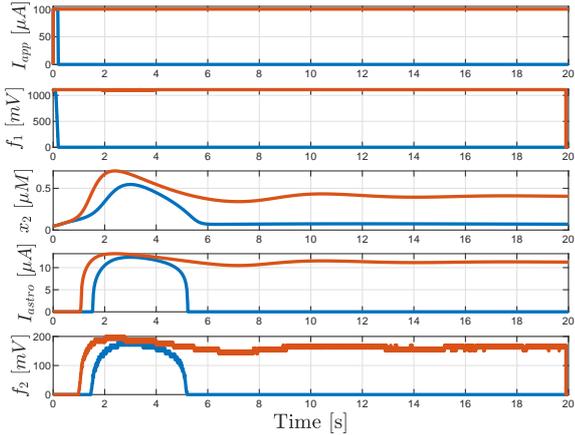}
    \caption[Numerical results for tripartite synapse]{The input and response of the tripartite synapse exposed: the short input, e.g. stimulation, signal (red), and a persistent signal (blue).}
    \label{fig:sim_tripartite}
\end{figure}
\section{Implications of the results for Working Memory: Numerical validation} \label{sec:4}
The application of a neuron-astrocyte network model is tested by expanding the tripartite synapse  - as described in section \ref{sec:2} - into a large-scale network representing the WM. For these purposes, a dual-layer structure is built by neuronal synapses, astrocytic gap junctions to connect astrocytes with each other, and the tripartite synapses as the connection between the neuronal and the astrocytic layer (Fig.~ \ref{fig:dual_layer}). We connected 1296 neurons consisting of $80\%$ excitatory and $20\%$ inhibitory cells following an exponential distribution depending on the distance between them. Four neurons are linked to the spatially closest astrocyte, respectively. 
The inter-astrocyte connections are assumed bidirectional, i.e., via (electrical) synapses (gap junctions) \citep{gordleeva2021modeling}.

\begin{figure}[h]
    \centering
    \includegraphics[width = 0.65\linewidth]{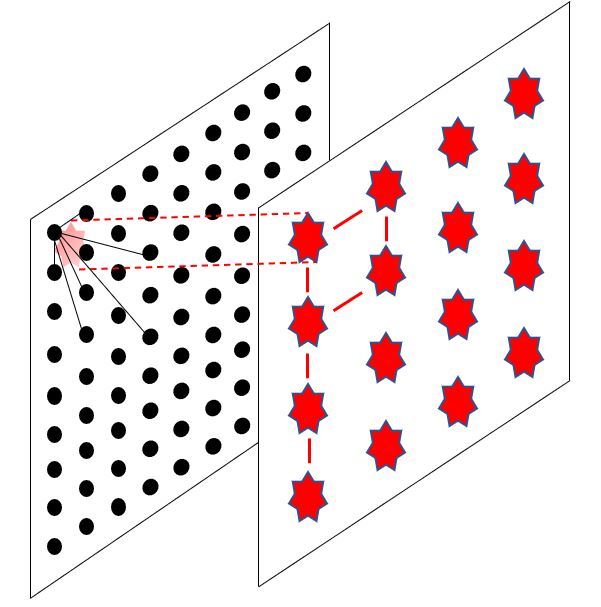}
    \caption[Dual layer network structure]{Dual layer network: The neuronal layer (left) and the astrocytic layer (right) form spatial interconnections (adapted from \citet{gordleeva2021modeling}).}
    \label{fig:dual_layer}
\end{figure}
As mentioned in the introduction, astrocytic gliotransmission and its effect on the postsynaptic neuron is considered a possible mechanism for WM. The simulation sequence reflects real experiments having a sequence of {\em stimulation}, {\em delay period}, and {\em recall}. In order to enhance the recall, a form of original stimulus can be re-introduced to the network which is called a cue. We assume that the selected target neurons receive a stimulating input for the simulation period $t_{stim} = 0.2 s$ which is followed by a delay period $t_{delay} = 2.8 s$. Next, the recall period starts which is accompanied by a weak and noisy recall cue. Simulation results for strong, weak, and no gliotransmission show the following performance. Strong gliotransmission (Fig.~\ref{fig:result_protocol1}) results in persistent enhanced activity of target neurons with an average frequency of approx. 150~Hz compared to non-target neurons at 10 to 20~Hz. In contrast to the other two scenarios, no recall cues are applied here. Although the difference between firing frequencies is strongly enhanced in this simulation, the general effect corresponds to the observed persistent activity of tuned neurons in \citet{funahashi1989mnemonic}.

\begin{figure}[h]
    \centering
    \includegraphics[width = 0.99\linewidth]{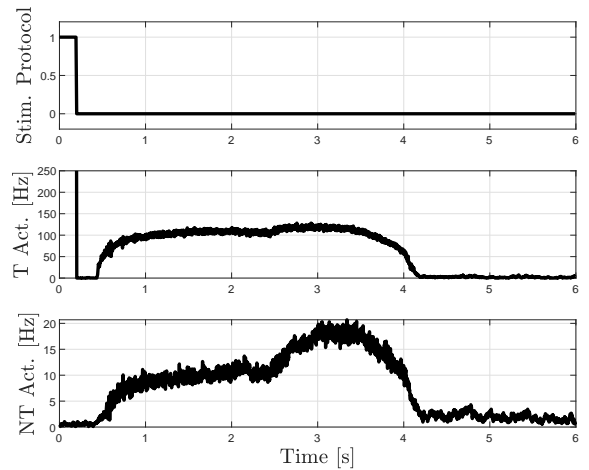}
    \caption[Simulation results - Case 2]{Strong gliotransmission (shot period): The activity of target neurons (T) is significantly higher than non-target neurons (NT).}
    \label{fig:result_protocol1}
\end{figure}

The effect of weak gliotransmission is significant, especially in the second subplot of Fig.~\ref{fig:result_protocol2}. While there is basically no difference in neuronal activity between target and non-target neurons during the delay period, the recall cue shows a significant difference in neuronal reactions. The generally low neuronal activity during the delay period also links to the hypothesis of sparse neuronal activity.

\begin{figure}[h]
  \centering
  \includegraphics[width = 0.99\linewidth]{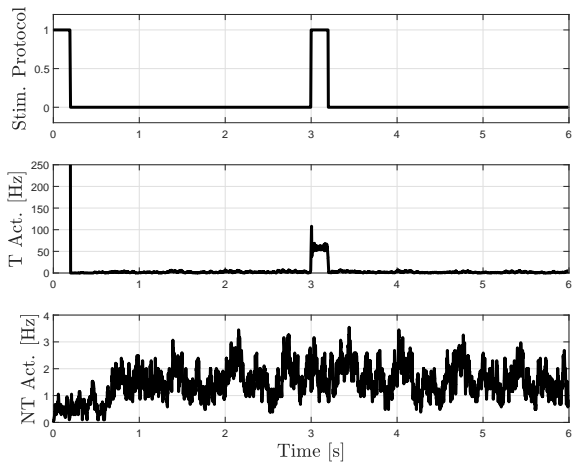}
   \caption[Simulation results - Case 3]{Weak gliotransmission: The activity of target neurons (T) is significantly higher than non-target neurons (NT) for the second cue. Here, the astrocytic gliotransmission is not strong enough to excite the neurons on its own but in combination with a second, weak cue signal.}
 \label{fig:result_protocol2}
\end{figure}

In the last scenario, we removed the effect of gliotransmission in order to examine whether the above-seen WM performance can purely be assigned to astrocytic gliotransmission. The absence of gliotransmission removes every difference in neuronal activity between target and non-target cells which can be translated to a totally dysfunctional WM model (Fig.~\ref{fig:result_protocol2_b}).

\begin{figure}[h]
    \centering
    \includegraphics[width = 0.99\linewidth]{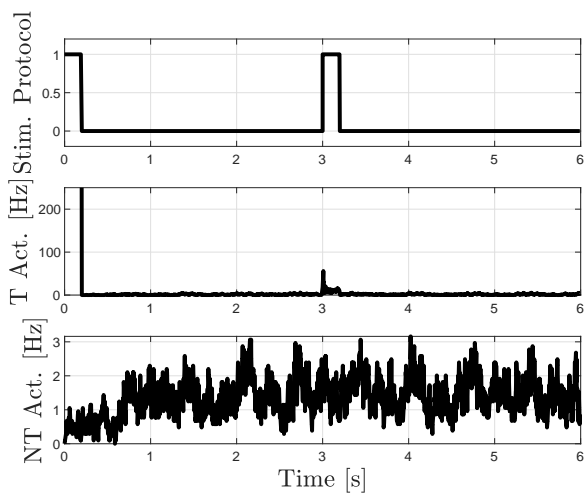}
    \caption[Simulation results - No gliotransmission]{No gliotransmission: The WM task shows sparse activity overall time and no significant recall which indicates poor memory performance.}
    \label{fig:result_protocol2_b}
\end{figure}

\section{Conclusion}\label{sec:5}
Using nonlinear stability analysis for a biologically detailed astrocyte model, we showed that the astrocyte response to a bounded input is bounded and we quantified this bound. Our analysis showed that astrocytic gliotransmission provides a stable activation of the postsynaptic neuron depending on the stimulus applied to the presynaptic neurons. This result indicates the possibility of storing presynaptic neuronal activity. We numerically verified the implication of this result to explain that both persistent and sparse activity of WM can be explained by our analysis. Further research, both in experimental as well as computational aspects, is necessary to examine the possible role of astrocytes in WM tasks. 
\bibliography{bib}             

\appendix
\section{Parameters of the Tripartite Synapse} \label{sec:A1}

Table~\ref{tab:neuron}: Parameters for the neuron model by \citet{izhikevich2003simple} following \citet{gordleeva2021modeling}.

\begin{table}[h]
    \centering
    \begin{tabular}{c|c}
    \textbf{Parameter} &  \textbf{Value} \\
    \hline \hline
    a &  0.1 \\
    \hline 
    b &  0.2 \\
    \hline 
    c &  -65mV \\
    \hline 
    d &  2 \\
    \hline
    $\eta_{syn}$ &  0.025 \\
    \hline
    $E_{syn,E}$ &  0 mV \\
    \hline
    $E_{syn,I}$ &  -90 mV \\
    \hline
    $k_{syn}$ &  0.2 mV
    \end{tabular}
    \caption[Overview of neuronal parameters]{}
    \label{tab:neuron}
\end{table}

Table~\ref{tab:n2}: Parameters for the release dynamics for glutamate dynamics \citep{gordleeva2012bi}.

\begin{table}[h]
    \centering
    \begin{tabular}{c|c}
    \textbf{Parameter}  \textbf{Value} \\
    \hline \hline
    $\alpha_{glu}$ & 10 $s^{-1}$  \\
    \hline 
    $k_{glu}$ & 600 $\mu$M $s^{-1}$ \\
    \hline 
    $A_{glu}$ & 5 $\mu$M $s^{-1}$ \\
    \hline 
    $G_{thr}$ & 0.7 
    \end{tabular}
    \caption[Overview of glutamate release parameters]{}
    \label{tab:n2}
\end{table} 

Table~\ref{tab:n3}: Parameters for the astrocyte model by \citet{li1994equations} with adaptions by \citet{nadkarni2003spontaneous} and \citet{ullah2006anti}. 

\begin{table}[h]
    \centering
    \begin{tabular}{c|c}
    \textbf{Parameter} &  \textbf{Value} \\
    \hline \hline
    $\frac{1}{\tau_{IP_3}}$ &  0.14 $s^{-1}$ \\
    \hline
    $c_0$ & 2.0 $\mu$M \\
    \hline 
    $a_2$ & 0.14 $\mu M^{-1} s^{-1}$\\
    \hline
    $c_1$ & 0.185 \\
    \hline
    $d_1$ & 0.13 $\mu$M \\
    \hline
    $d_2$ & 1.049 $\mu$M \\
    \hline
    $d_3$ & 943.4 nM \\
    \hline
    $d_5$ & 82 nM \\
    \hline
    $d_{Ca}$ & 0.05 $s^{-1}$ \\
    \hline
    $d_{IP3}$ & 0.1 $s^{-1}$ \\ 
    \hline
    $IP_3^{*}$ & 0.16 $\mu$M \\
    \hline 
    $k_1$ & 0.5 $s^{-1}$ \\
    \hline
    $k_2$ & 1 $\mu$M \\
    \hline
    $k_3$ & 0.1 $\mu$M \\
    \hline
    $k_4$ & 1.1 $\mu$M \\
    \hline
    $v_1$ & 6 $s^{-1}$ \\
    \hline 
    $v_2$ & 0.11 $s^{-1}$ \\
    \hline
    $v_3$ & 2.2 $\mu$M $s^{-1}$ \\
    \hline
    $v_4$ & 0.3 $\mu$M $s^{-1}$ \\
    \hline
    $v_6$ & 0.2 $\mu$M $s^{-1}$ \\
    \hline
    $\alpha$ & 0.8
    \end{tabular}
    \caption[Overview of astrocyte parameters]{}
    \label{tab:n3}
\end{table} 

Table~\ref{tab:network}: Parameters for the large-scale neuron-astrocyte WM network.

\begin{table}[h]
    \centering
    \begin{tabular}{c|c}
    \textbf{Parameter} &  \textbf{Value} \\
    \hline \hline
    $N_E$ & 1296 \\
    \hline
    $N_A$ & 324 \\
    \hline
    $N_{E,Syn}$ & 28 \\
    \hline
    $\lambda$ & 5 \\
    \hline
    $N_{A,Syn,min}$ & 2 \\
    \hline
    $N_{A,Syn,max}$ & 4 \\
    \hline
    $r_{E-I}$ & 4
    \end{tabular}
    \caption[Overview of network parameters]{}
    \label{tab:network}
\end{table} 

\section{Derivation of the Extended Astrocyte Model} \label{sec:A2} 

The first step in order to obtain the extended astrocyte model is the continuous estimation of $I_{astro}$. The original definition is given by \citet{nadkarni2003spontaneous}
\begin{equation*}
    I_{x_2}=2.11 \Theta(\ln y) \ln y, \quad y=x_2 / \mathrm{nM} -196.69.
\end{equation*}

Instead, we derive a continuously differentiable estimation by fitting the function 
\begin{equation*}
    f(x_2) = a \tanh \left(b x_2 + c\right) + d)
\end{equation*} 
in a realistic range of $[Ca^{2+}]$ values $x_2 \in [0.05, 0.7]$ $\mu M$ obtained via numerical simulations. The parameter fitting for $a = 6.3611$, $b = 14.682$, $c = -3.3582$ and $d = 6.3611$ is done via Matlab. As it can be seen in Fig.~\ref{fig:I_astro}, the qualitative development of $I_{astro}$ is captured although a considerable error is visible at the onset due to the smoothing process.

\begin{figure}[h]
    \centering
    \includegraphics[width = 0.99\linewidth]{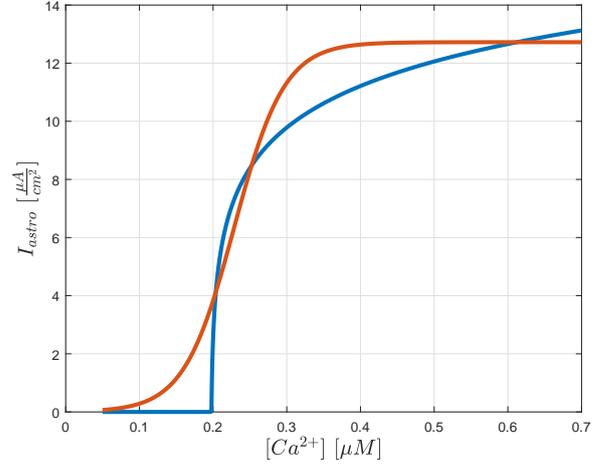}
    \caption[Estimation of $I_{astro}$]{Estimation of $I_{astro}$: The relation between the $Ca^{2+}$ concentration and the astrocytic current is estimated by a $\tanh$-function.}
    \label{fig:I_astro}
\end{figure}

Since $x_4$ should stabilize at zero for $I_{astro}=0$ and otherwise should converge to the firing frequency of the Izhikevich model, we propose the following equation:

\begin{equation*}
    \dot{x}_4 = -x_4 + f(I_{astro}).
\end{equation*}

Considering the sudden onset of the firing, we upgrade the model as follow:

\begin{equation*}
    \dot{x}_4 = -x_4 + 0.5 (tanh(I-I_{thr})+1) f_1(I_{astro}),
\end{equation*}

where the heaviside function $\Theta$ is approximated by a $\tanh$-function and the threshold value $I_{thr} = 3.9 \mu A$. Additionally, considering a linear relation between the input current and the firing frequency, we have

\begin{equation*}
    f_1(I_{astro}) \approx p_1 I_{astro} + p_2
\end{equation*}

with $p_1 = 16.82$ and $p_2 = -40.29$ via the Linear Least Square Method. Lastly, the tuning parameter $\eta$ is introduced in order to model weak ($\eta = 25\%$) or strong astrocytic gliotransmission ($\eta = 100\%$) which gives the following model:

\begin{equation*}
    \dot{x}_4 = -x_4 + 0.5 (\tanh(\eta I_{astro} -I_{thr})+1) (p_1 \eta I_{astro} + p_2).
\end{equation*}

\section{Equilibrium and Eigenvalues} \label{sec:A3} 
The unperturbed, $u=0$, equilibrium of the system is calculated: $x_0 = \begin{bmatrix}
0.6858 & 0.06612 & 0.8882
\end{bmatrix}$. The corresponding eigenvalues of the linearized model are obtained as: $$ \begin{bmatrix}
-4.2324 & -0.12 + 0.023i & -0.12 - 0.023i
\end{bmatrix}.$$ \\
The system's equilibrium under the influence of with the positive input $u=A_{glu}=5$ is $x_0 = \begin{bmatrix}
36.77  & 0.4061 & 0.7165
\end{bmatrix}.$ The corresponding eigenvalues of the linearized model are obtained as: $$ \begin{bmatrix}
-0.14  & -0.27 + 0.89i & -0.27 - 0.89i
\end{bmatrix}.$$








\end{document}